\documentclass[11pt]{article}

\usepackage[utf8]{inputenc} 
\usepackage[T1]{fontenc}
\usepackage{amsmath, amssymb, amsthm}
\usepackage{bm,bbm}
\usepackage{mathtools}
\usepackage{lmodern}
\usepackage{commath}
\usepackage{subfig}
\usepackage{graphicx}
\usepackage{indentfirst}
\usepackage{hyperref}
\usepackage{enumerate}
\usepackage{makecell}
\usepackage{fancyhdr}
\usepackage{setspace}
\usepackage[margin=2.5cm]{geometry}
\usepackage{tabularx}
\usepackage[title]{appendix}

\newtheorem{Theorem}{Theorem}
\newtheorem{Lemma}{Lemma}%
\newtheorem{Proposition}{Proposition}%
\newtheorem{Example}{Example}%

\theoremstyle{definition}

\DeclareMathOperator{\Exp}{E}
\DeclareMathOperator{\Var}{Var}
\DeclareMathOperator{\Covar}{Cov}
\DeclareMathOperator{\DKL}{D_\mathrm{KL}}
\DeclareMathOperator{\vol}{vol}
\DeclareMathOperator{\Info}{I}
\DeclareMathOperator{\Ent}{H}
\DeclareMathOperator{\ent}{h}
\DeclareMathOperator{\Det}{det}
\DeclareMathOperator{\covol}{covol}

\newcommand{\calB}{\mathcal{B}}
\newcommand{\calD}{\mathcal{D}}

\newcommand{\calP}{\mathcal{P}}
\newcommand{\calQ}{\mathcal{Q}}
\newcommand{\calV}{\mathcal{V}}
\newcommand{\calX}{\mathcal{X}}
\newcommand{\calY}{\mathcal{Y}}
\newcommand{\R}{\mathbb{R}}

\newcommand{\Z}{\mathbb{Z}}

\newcommand{\rv}[1]{{#1}}

\makeatletter
\let\@fnsymbol\@arabic
\makeatother
\newcommand*\samethanks[1][\value{footnote}]{\footnotemark[#1]}

\title{\textbf{Information Properties of a Random Variable Decomposition through Lattices}}
\author{Fábio~C.~C.~Meneghetti\thanks{
	The authors are with the Institute of Mathematics, Statistics and Scientific Computing~(IMECC), University of Campinas~(Unicamp), Brazil. E-mail: \href{mailto:fabiom@ime.unicamp.br}{\texttt{fabiom@ime.unicamp.br}}, \href{mailto:hmiyamoto@ime.unicamp.br}{\texttt{hmiyamoto@ime.unicamp.br}}, \href{mailto:sueli@unicamp.br}{\texttt{sueli@unicamp.br}}.
	}
	\and Henrique~K.~Miyamoto\samethanks \and Sueli~I.~R.~Costa\samethanks}
\date{\vspace{-1.5em}}

\begin{document}

\maketitle

\begin{abstract}
	A full-rank lattice in the Euclidean space is a discrete set formed by all integer linear combinations of a basis. Given a probability distribution on $\mathbb{R}^n$, two operations can be induced by considering the quotient of the space by such a lattice: wrapping and quantization. For a lattice $\Lambda$, and a fundamental domain $\calD$ which tiles $\mathbb{R}^n$ through $\Lambda$, the wrapped distribution over the quotient is obtained by summing the density over each coset, while the quantized distribution over the lattice is defined by integrating over each fundamental domain translation. These operations define wrapped and quantized random variables over $\calD$ and $\Lambda$, respectively, which sum up to the original random variable. We investigate information-theoretic properties of this decomposition, such as entropy, mutual information and the Fisher information matrix, and show that it naturally generalizes to the more abstract context of locally compact topological groups.
	
	\vspace{1em}
	
\end{abstract}

\section{Introduction}

Lattices are discrete sets in $\mathbb{R}^n$ formed by all integer linear combinations of a set of independent vectors, and have found different applications, such as in information theory and communications~\cite{conway1999,zamir2014,costa2017}. Given a probability distribution in $\R^n$, two operations can be induced by considering the quotient of the space by a lattice: wrapping and quantization.

The wrapped distribution over the quotient is obtained by summing the probability density over each coset. It is used to define parameters for lattice coset coding, particularly for the AWGN and wiretap channels, such as the flatness factor, which is, up to a constant, the $L^\infty$ distance from a wrapped probability distribution to a uniform one~\cite{ling2014,damir2021}. This factor is equivalent to the smoothing parameter, used in post-quantum lattice-based cryptography~\cite{chung2013}. In the context of directional statistics, wrapping has been used as a standard way to construct distributions on a circle and on a torus~\cite{mardia2000}.

The quantized distribution over the lattice can be defined by integrating over each fundamental domain translation, thus corresponding to the distribution of the fundamental domains after lattice-based quantization is applied. Lattice quantization has different uses in signal processing and coding: for instance, it can achieve the optimal rate-distortion trade-off and can be used for shaping in channel coding~\cite{zamir2014}. A special case of interest is when the distribution on the fundamental region is uniform, which amounts to high-resolution quantization or dithered quantization~\cite{zamir1996,ling2013}.

In this work, we relate these two operations by remarking that the random variables induced by wrapping and quantization sum up to the original one. We study information properties of this decomposition, both from classical information theory~\cite{cover2006} and from information geometry~\cite{amari2000}, and provide some examples for the exponential and Gaussian distributions. We also propose a generalization of these ideas to locally compact groups. Probability distributions on these groups have been studied in \cite{heyer1977}, and some information-theoretic properties have been investigated in \cite{chirikjian2009, johnson2000, chirikjian2010}. In addition to probability measures, one can also define the notions of lattice and fundamental domains on them, thereby generalizing the Euclidean case. We show that wrapping and quantization are also well defined, and provide some illustrative examples.

\section{Lattices, Wrapping and Quantization} \label{sec:Lattices}

\subsection{Lattices and Fundamental Domains}

A lattice $\Lambda$ in $\R^n$ is a discrete additive subgroup of $\R^n$, or, equivalently, the set $\Lambda = \set{\alpha_1 b_1 + \dots + \alpha_k b_k \;\middle|\; \alpha_1,\dots, \alpha_k \in \Z}$ formed by all integer linear combinations of a set of linearly independent vector $\set{b_1, \dots, b_k} \subset \R^n$, called a \emph{basis} of $\Lambda$. A matrix $B$ whose column vectors forms a basis is called a \emph{generator matrix} of $\Lambda$, and we have $\Lambda = B \Z^k$. The lattice dimension is $k$, and, if $k=n$, the lattice is said to be \emph{full-rank}; we henceforth consider full-rank lattices.
A lattice $\Lambda$ defines an equivalence relation in $\R^n$: $x \sim y \iff x - y \in \Lambda$. The associated equivalence classes are denoted by $\bar x$ or $x+\Lambda$. The set of all equivalence classes is the lattice quotient ${\R^n}/{\Lambda}$, and we denote the standard projection $\pi \colon \R^n \to \R^n / \Lambda, \ \pi(x) = \bar x$.

Let $\calD$ be a Lebesgue-measurable set of $\R^n$ and $\Lambda$ a lattice. We say that $\calD$ is a \emph{fundamental domain} or a \emph{fundamental region} of $\Lambda$, or that $\calD$ \emph{tiles} $\R^n$ by $\Lambda$, if\footnote{
	It is often only asked that the intersection in item~2) has Lebesgue measure zero, but we require it to be empty.
} 1)~$\bigcup_{\lambda \in \Lambda} (\lambda + \calD) = \R^n$, and 2)~$(\lambda+\calD) \cap (\tilde\lambda+\calD) = \emptyset$, for all $\lambda \neq \tilde\lambda$ in $\Lambda$. Given a fundamental domain $\calD$, each coset $\bar x \in \R^n / \Lambda$ has a unique representative in $\calD$, i.e., the measurable map $\pi|_\calD \colon \calD \to \R^n / \Lambda$ is a bijection. This fact suggests using a fundamental domain to represent the quotient. Each fundamental domain contains exactly one lattice point, which may be chosen as the origin.
One example of fundamental domain is the \emph{fundamental parallelotope} with respect to a basis $\set{b_1, \dots, b_n}$, namely $\calP(\Lambda) \coloneqq \set{x=\alpha_1 b_1 + \dots + \alpha_n b_n \;\middle|\; \alpha_1, \dots, \alpha_n \in \left[0,1\right[}$. Another one is the Voronoi region~$\calV (\Lambda)$ of the origin, given by the points that are closer to the origin than to any other lattice point, with an appropriate choice for ties. It is a well-known fact that every fundamental domain has the same volume, denoted by $\covol\Lambda \coloneqq \vol\calD = \abs{\Det B}$, for any generator matrix $B$ of $\Lambda$.

\subsection{Wrapping and Quantization}

Consider $\R^n$ with the Lebesgue measure $\mu$, and $P$ a probability measure such that $P \ll \mu$. Then the probability density function~(pdf) of $P$ is $p = \od{P}{\mu}$, the Radon-Nikodym derivative. For fixed full-rank lattice $\Lambda$ and fundamental domain $\calD$, the \emph{wrapping} of $P$ by $\Lambda$ is the distribution $P_\pi \coloneqq \pi_* P$ on $\R^n / \Lambda$, given by $P_\pi(A) = P (\pi^{-1}A)$. For simplicity, we identify $\R^n / \Lambda$ with $\calD$ to regard $P_\pi$ as a distribution over $\calD$, and then we have $\pi\colon \R^n \to \calD$ given by $(y+\lambda) \mapsto y$, for all $y \in \calD, \lambda\in\Lambda$. Using this identification, the wrapping has density $p_\pi = \od{P_\pi}{\mu}$ given by

\begin{equation}
	p_\pi (y) = \sum_{\lambda \in \Lambda} p(y+\lambda).
\end{equation}

A construction that is, in some sense, dual to wrapping is quantization. Note that each fundamental domain $\calD$ partitions the space as $\R^n  = \bigsqcup_{\lambda \in \Lambda} (\lambda + \calD)$. The quantization function is the measurable map $\calQ \colon \R^n \to \Lambda$, given by $(y + \lambda) \mapsto \lambda$, for $y \in \calD$ and $\lambda \in \Lambda$. The \emph{quantized probability distribution} of $P$ on the discrete set $\Lambda$ is $P_\calQ \coloneqq \calQ_*P$, given by $P_\calQ (A) \coloneqq P (\calQ^{-1} A)$. The probability mass function of the quantized distribution is then

\begin{equation}
	p_\calQ (\lambda) = \int_{\calD} p(y+\lambda) \dif x. 
\end{equation}

Letting $\rv{X}$ be a vector random variable in $\R^n$ with distribution $p$, we define $\rv{X}_\pi \coloneqq \pi(\rv{X})$ and $\rv{X}_\calQ \coloneqq \calQ(\rv{X})$ the wrapped and quantized random variables, respectively. By definition, they are distributed according to $p_\pi$ and $p_\calQ$. Interestingly, they sum up to the original one:

\begin{equation} \label{eq:decomposition}
	\rv{X} = \rv{X}_\pi + \rv{X}_\calQ,
\end{equation}
since $\pi + \calQ = \mathrm{id}_{\R^n}$. Note also that $X_\pi + X_\calQ$ has the same distribution as $(X_\pi, X_\calQ)$, by the bimeasurable bijection $y+\lambda \mapsto (y,\lambda)$. These factors, however, are not independent, since, in general, $p(y+\lambda) \neq p_\pi(y) p_\calQ (\lambda)$. The difference between $p(x)$ and $(p_\pi \otimes p_\calQ)(x) \coloneqq p_\pi \del{\pi(x)}p_\calQ \del{\calQ(x)}$ shall be illustrated in the following examples.
Note that the expression for the quantized distribution depends on the choice of fundamental domain, while the wrapped distribution does not, up to a lattice translation.

We say a random variable $X$ over $[0,\infty)$ is \emph{memoryless} if $\bar C(t) = \bar C(t+s)/ \bar C(s)$ for all $t,s$, where $\bar C(t) \coloneqq P[X>t]$ is the tail distribution function. In particular, a memoryless distribution satisfies $\bar C(y+\lambda) = \bar C(y) \bar C(\lambda)$ for all $y\in\calD$, $\lambda\in\Lambda$, which implies $p = p_\pi \otimes p_\calQ$. The converse, however, is not true; for example, independence holds whenever $p$ is constant on each region $\lambda+\calD$, for $\lambda \in \Lambda$.

\begin{Example}  \label{ex:exponential}
	The exponential distribution, parametrized by $\nu>0$, is defined as $p(x) = \nu e^{-\nu x} \mathbbm{1}_{\left[0,+\infty\right[}(x)$, where $\mathbbm{1}_A(x)$ takes value $1$ if $x\in A$, and 0 otherwise.
	Choosing the lattice $\Lambda = \alpha \Z$, $\alpha \in \R_+$, and the fundamental domain $\calD = \left[0,\alpha\right[$, one can write closed-form expressions for the wrapped and quantized distributions:
	\begin{equation}
		p_\pi(y) = \frac{\nu e^{-\nu y}}{1 - e^{-\nu \alpha}}, \quad  y \in \calD
		\qquad \text{and} \qquad
		p_\calQ(\lambda) = e^{-\nu \lambda} \left( 1 - e^{-\nu \alpha} \right), \quad  \lambda \in \Lambda \cap \R_+.
	\end{equation}
	Note that, in this special case, $p = p_\pi \otimes p_\calQ$, as a consequence of memorylessness. The wrapped distribution with $\alpha=2\pi$, which amounts to a distribution on the unitary circle, is well studied in~\cite{jammalamadaka2004}.
\end{Example}

\begin{Example} \label{ex:gaussian}
	Consider the univariate Gaussian distribution $p(x) = (2\pi \sigma^2)^{-1/2} \exp (-(x-\mu)^2 / 2 \sigma^2)$ and the lattice $\Lambda = \alpha \Z$, with fundamental domain $\calD = \left[ -\frac{\alpha}{2}, \frac{\alpha}{2} \right[$, $\alpha \in \R_+$. The wrapped and quantized distributions are given respectively by
	\begin{equation*}
		p_\pi(y) =
		\frac{1}{\sqrt{2 \pi \sigma^2}} \sum_{i \in \Z} e^{-\frac{ (y-\mu+\alpha i)^2 }{2\sigma^2}}, \quad  y \in \calD
		\qquad \text{and} \qquad
		p_\calQ(\lambda) =
		\frac{1}{\sqrt{2 \pi \sigma^2}} \int_{\lambda - \frac{\alpha}{2}}^{\lambda + \frac{\alpha}{2}} e^{-\frac{(x - \mu)^2}{2 \sigma^2}} \dif x, \quad  \lambda \in \Lambda.
	\end{equation*}
	The value $\alpha=2\pi$ for the wrapped distribution on a unitary circle is usually considered in directional statistics~\cite{mardia2000}. Figure~\ref{fig:example_gaussian} illustrates the original, wrapped, quantized and product distributions for different zero-mean Gaussian distributions. As it can be seen in the figure, in this case, $p(x) \neq p_\pi(y)p_\calQ(\lambda)$.
	
	\begin{figure}
		\centering
		\subfloat[Original.]{\includegraphics[width=.22\textwidth]{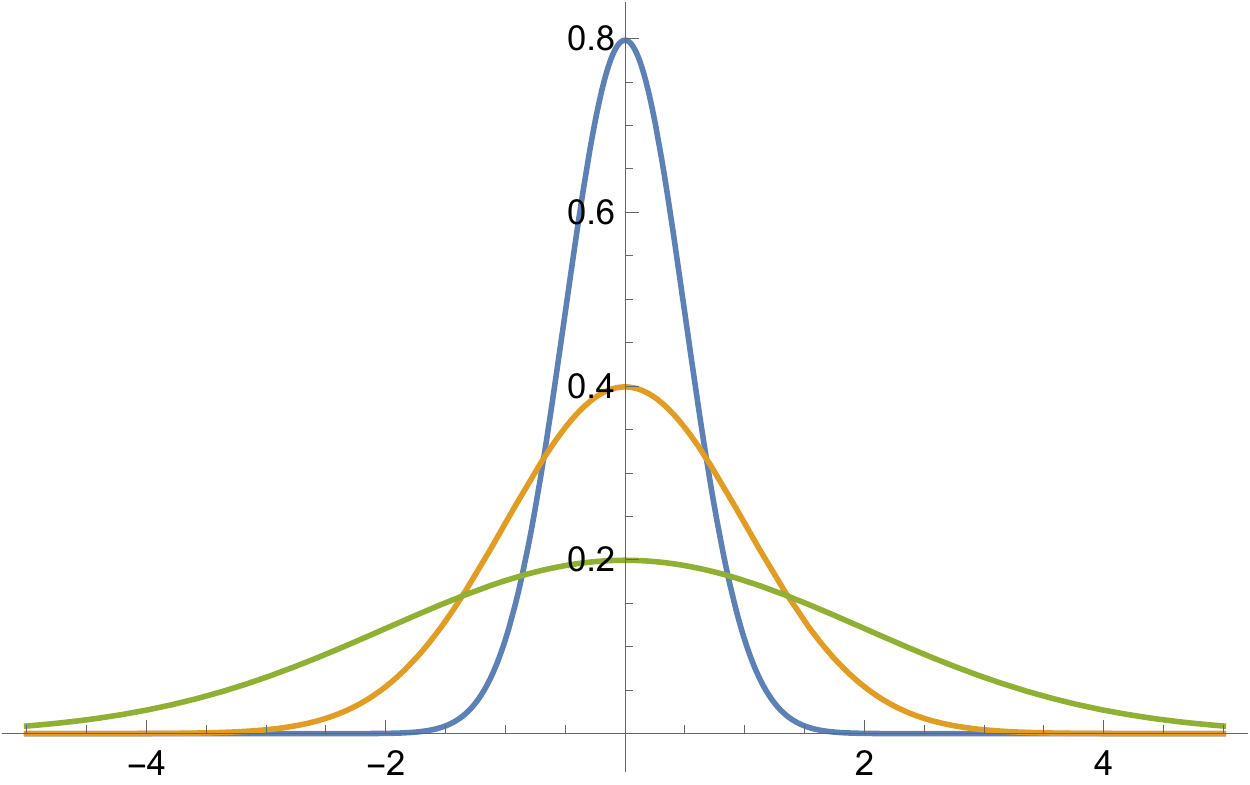}} \hspace{1em}
		\subfloat[Wrapped.]{\includegraphics[width=.22\textwidth]{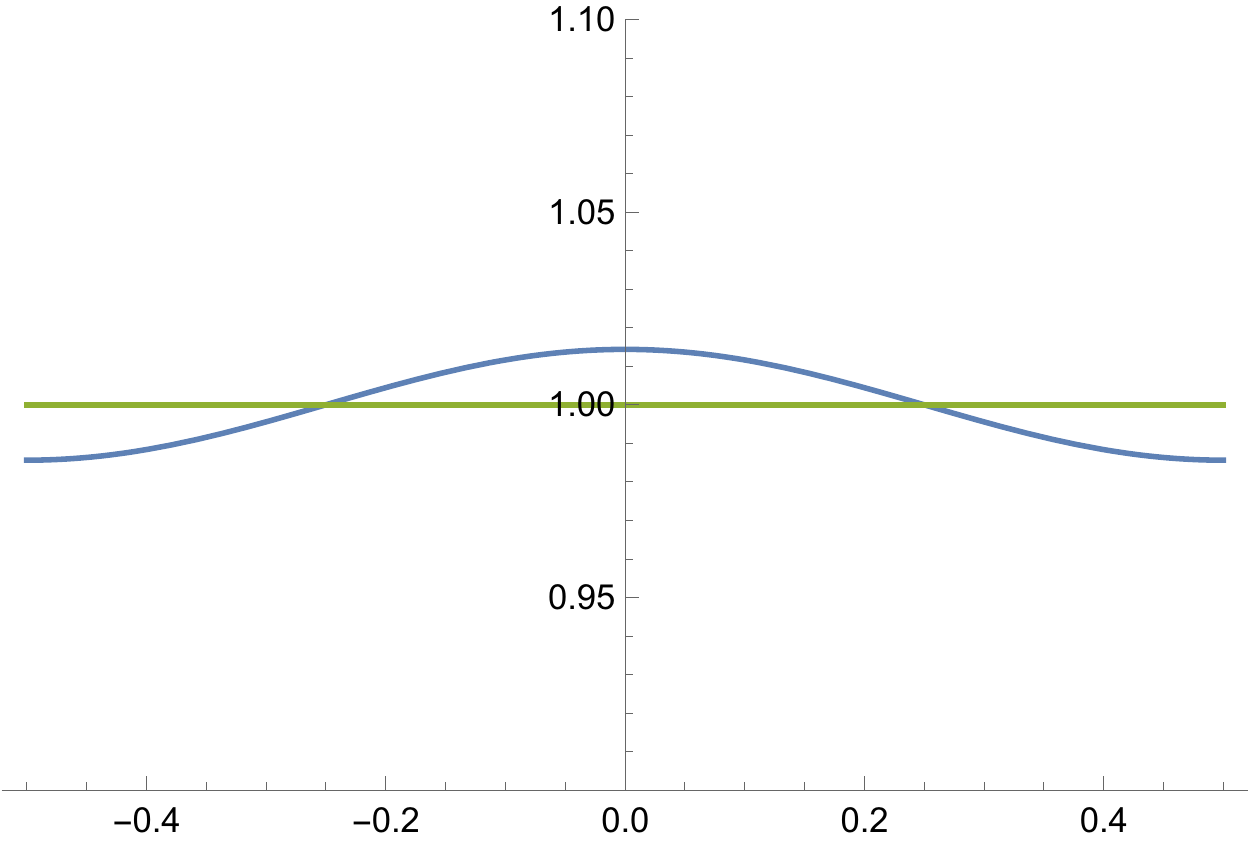}}
		\subfloat[Quantized.]{\includegraphics[width=.22\textwidth]{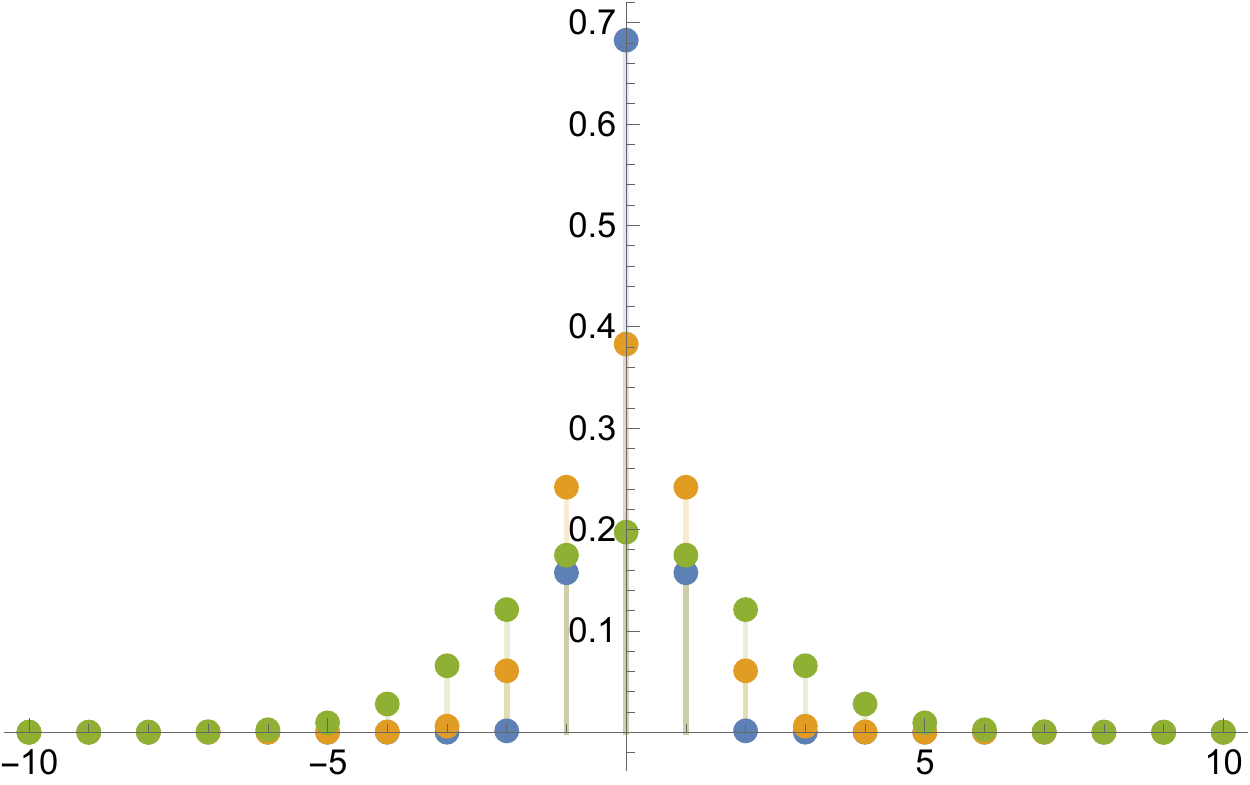}} \hspace{1em}
		\subfloat[Product.]{\includegraphics[width=.22\textwidth]{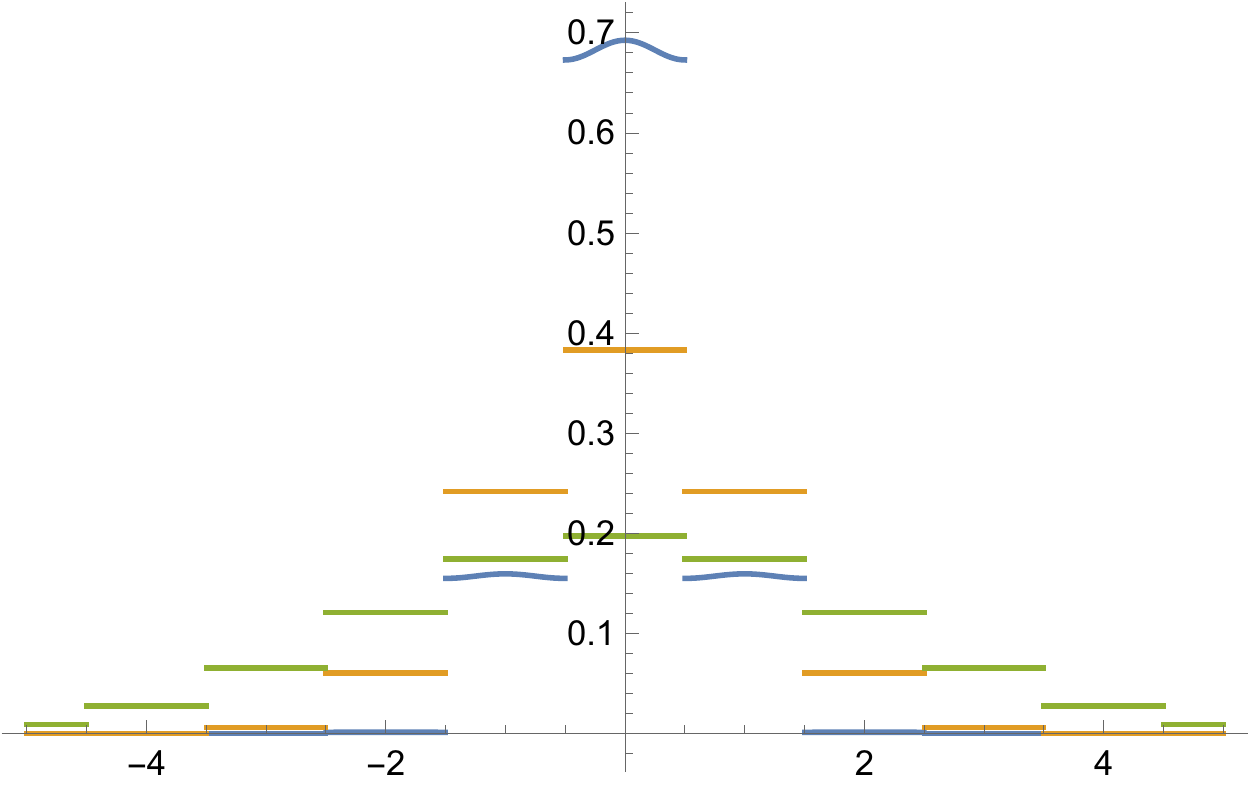}}
		\caption{Example of zero-mean Gaussian distributions and their corresponding wrapped, quantized and product distributions, with $\Lambda = \Z$ and $\calD = [-\frac{1}{2}, \frac{1}{2}[$ for different variances: $\sigma^2 = 0.25$ (blue), $\sigma^2 = 1$ (orange), $\sigma^2 = 4$ (green).}
		\label{fig:example_gaussian}
	\end{figure}
	
\end{Example}

A straightforward consequence of the decomposition \eqref{eq:decomposition} is
\begin{enumerate}
	\item\label{item:exp-decomp} $\Exp[\rv{X}] = \Exp[\rv{X}_\pi] + \Exp[\rv{X}_\calQ]$;
	\item $\Var[\rv{X}] = \Var[\rv{X}_\pi] + \Var[\rv{X}_\calQ] + \Covar[\rv{X}_\pi, \rv{X}_\calQ] + \Covar[\rv{X}_\calQ, \rv{X}_\pi]$,
\end{enumerate}
where $\Exp[\cdot]$, $\Var[\cdot]$ and $\Covar[\cdot,\cdot]$ denote respectively the expectation, the variance and the cross-covariance operators.

We note that different types of discretization have also been studied, other then integrating over a fundamental domain~\cite{chakraborty2015}. For instance, in~\cite{ling2014,nielsen2022,luzzi2018} the discretized distribution is defined by restricting the original pdf $p(x)$ to the lattice $\Lambda$, and then normalizing:

\begin{equation}
	D_{\Lambda,c} (\lambda)
	\coloneqq
	\frac{p(c+\lambda)}{\sum_{\tilde\lambda \in \Lambda} p(c+\tilde\lambda)},
\end{equation}
for a fixed $c \in \calD$. This discretization is nothing other than the conditional distribution of $\rv{X}_\calQ$ given that $\rv{X}_\pi = c$, expressed as $p_{\calQ | \pi} (\lambda | c) = {p(c+\lambda)}/{p_\pi(c)}$. Moreover, when $p = p_\pi \otimes p_\calQ$, such as in the exponential distribution, cf. Example \ref{ex:exponential}, then $D_{\Lambda,c} (\lambda) = p_\calQ (\lambda)$.

\section{Information Properties} \label{sec:InformationProp}

\subsection{Information-theoretic Measures}

Let us consider a random variable $\rv{X}$ with distribution $p$ and the induced wrapped and quantized ones, respectively, $\rv{X}_\pi \sim p_\pi$ and $\rv{X}_\calQ \sim p_\calQ$. The mutual information between $\rv{X}_\pi$ and $\rv{X}_\calQ$ is defined as the Kullback-Leibler divergence $\Info(\rv{X}_\pi; \rv{X}_\calQ) \coloneqq \DKL\left( p \| p_\pi \otimes p_\calQ \right)$,
and is a measure of how non-independent the marginal distributions $p_\pi$ and $p_\calQ$ are~\cite{cover2006}. Using the theorem of change of variables, we have
\begin{align}
	\Info(X_\pi; X_\calQ) 
	&= \Exp_X \left[ \log \textstyle\frac{p(X)}{ p_\pi \otimes p_\calQ (X) } \right] \nonumber\\
	&= \Exp_X [\log p(X)] - \Exp_X [\log p_\pi(X_\pi)] - \Exp_X [\log p_\calQ(X_\calQ) ]\nonumber\\
	&= \Exp_X [\log p(X)] - \Exp_{X_\pi} [\log p_\pi(X_\pi)] - \Exp_{X_\calQ} [\log p_\calQ(X_\calQ) ]\nonumber\\
	&= \ent(\rv{X}_\pi) + \Ent(\rv{X}_\calQ) - \ent(\rv{X}) \label{eq:MI-decomposition}.
\end{align}

\noindent
Note that, from this decomposition, we have $\ent(X) \le \ent(X_\pi) + \Ent(X_\calQ)$.

\begin{Proposition} \label{prop:mi-upper-bound}
	Let $\rv{X}$ be a random variable, and $\rv{X}_\pi$ and $\rv{X}_\calQ$ the respective wrapped and quantized random variables, using the lattice $\alpha\Z$. Denote $\mu_\calQ \coloneqq \Exp[\rv{X}_\calQ]$ and $\sigma_\calQ^2 \coloneqq \Var[\rv{X}_\calQ]$.
	If $\rv{X}$ has support $[0,\infty)$, then the mutual information~$\Info(\rv{X}_\pi; \rv{X}_\calQ)$ between $\rv{X}_\pi$ and $\rv{X}_\calQ$ is upper-bounded by
	\begin{equation} \label{eq:mi-bound-2}
		\Info(\rv{X}_\pi; \rv{X}_\calQ)
		< \log \left( e \left( \mu_\calQ + \alpha/2 \right) \right)
		- \ent(\rv{X}).
	\end{equation}
	If $\rv{X}$ has support $\R$, then~$\Info(\rv{X}_\pi; \rv{X}_\calQ)$ is upper-bounded by
	\begin{equation} \label{eq:mi-bound}
		\Info(\rv{X}_\pi; \rv{X}_\calQ)
		< \frac{1}{2}\log\del{ 2\pi e \sigma_\calQ^2}
		+ \frac{2 \log e}{\exp\del{2 \pi^2 \alpha^{-2} \sigma_\calQ^2} - 1}
		- \ent(\rv{X}).
	\end{equation}
\end{Proposition}

\begin{proof}
	First, $\ent(X_\pi) \le \log \alpha$, since the uniform distribution maximizes entropy on a bounded support.
	Then, note that the mean and variance of the integer-valued random variable $\alpha^{-1} X_\calQ$ are $\alpha^{-1} \mu_\calQ$ and $\alpha^{-2} \sigma_\calQ^2$, respectively.
	For~\eqref{eq:mi-bound-2}, use that, for positive integer random variables, $H(X_\calQ) < \log \left( e \left( \mu_\calQ/\alpha + 1/2 \right) \right)$, as in \cite[Thm.~8]{rioul2022};
	for~\eqref{eq:mi-bound}, the upper-bound for integer-valued random variables from \cite[Thm.~10]{rioul2022} gives us $\Ent(X_\calQ) < \frac{1}{2} \log\del{2\pi e \alpha^{-2} \sigma_\calQ^2} + \frac{2 \log e}{\exp\del{2 \pi^2 \alpha^{-2} \sigma_\calQ^2} - 1}$.
	Replacing the corresponding inequalities in \eqref{eq:MI-decomposition} yields the desired results.
\end{proof}

The following lemma can be found in \cite[Appendix~3]{zamir2014}.

\begin{Lemma}
\label{lem:wrapping-reduces-H}
	$\ent(X_\pi) \le \ent(X)$.
\end{Lemma}

\begin{proof}
	$\ent(X) = \ent(X_\pi) + \Ent(X_\calQ | X_\pi)$, and $\Ent(X_\calQ | X_\pi) \ge 0$, since it is a discrete entropy.
\end{proof}

\begin{Proposition} \label{prop:scaling}
	Let $\Lambda_\alpha \coloneqq \alpha \Lambda$, $\alpha>0$, be a family of lattices, with fundamental domains $\calD_\alpha \coloneqq \alpha \calD$.
	\begin{enumerate}
		\item If $\calD$ is connected, and $p$ is continuous and Riemann-integrable, then $\lim_{\alpha \to 0} \Info(X_\pi; X_\calQ) = 0$.
		\item If $0$ is an interior point of $\calD$, then $ \lim_{\alpha \to +\infty} \Info(X_\pi; X_\calQ) = 0$.
	\end{enumerate}
\end{Proposition}
\begin{proof}
	For $\alpha \to 0$, the proof is an adaptation of \cite[Thm.~8.3.1]{cover2006}. Since $\calD$ is connected and $p$ is continuous, we can use the mean value theorem: for every $\lambda \in \Lambda$ there exists a $x_{\lambda,\alpha} \in (\lambda+\calD_\alpha)$ such that $p(x_{\lambda,\alpha}) \vol{\calD_\alpha} = p_\calQ(\lambda)$. Therefore, we can write $\Ent(X_\calQ) = - \sum_{\lambda \in \Lambda_\alpha} p(x_{\lambda,\alpha}) \log\del{p(x_{\lambda,\alpha})} \vol\calD_\alpha - \log\del{\vol\calD_\alpha}$, using that $\sum_{\lambda \in \Lambda_\alpha} p(x_{\lambda,\alpha}) \vol\calD_\alpha = 1$. The first term is an $n$-dimensional Riemann sum, and converges to $\ent(X)$ when $\alpha \to 0$, while the second term gets arbitrarily small. Therefore, $0 \le \Info(X_\pi; X_\calQ) \le \Ent(X_\calQ) + \log\del{\vol\calD_\alpha} - \ent(X) \to 0$, so $\Info(X_\pi; X_\calQ) \to 0$.
	
	For $\alpha \to +\infty$, note that, from Lemma~\ref{lem:wrapping-reduces-H}, $\Info(X_\pi; X_\calQ) \le \Ent(X_\calQ)$. But, by choosing $\alpha$ sufficiently large, we can make $p_\calQ(0) = \int_{\calD_\alpha} p(x) \dif x$ arbitrarily close to $1$, since $0$ is in the interior of $\calD_\alpha$. Therefore, $\Ent(X_\calQ)$ can be made arbitrarily small.
\end{proof}

\begin{figure}
	\centering
	\subfloat[Exponential distributions.\label{fig:mi-exp}]{\includegraphics[width=.45\textwidth]{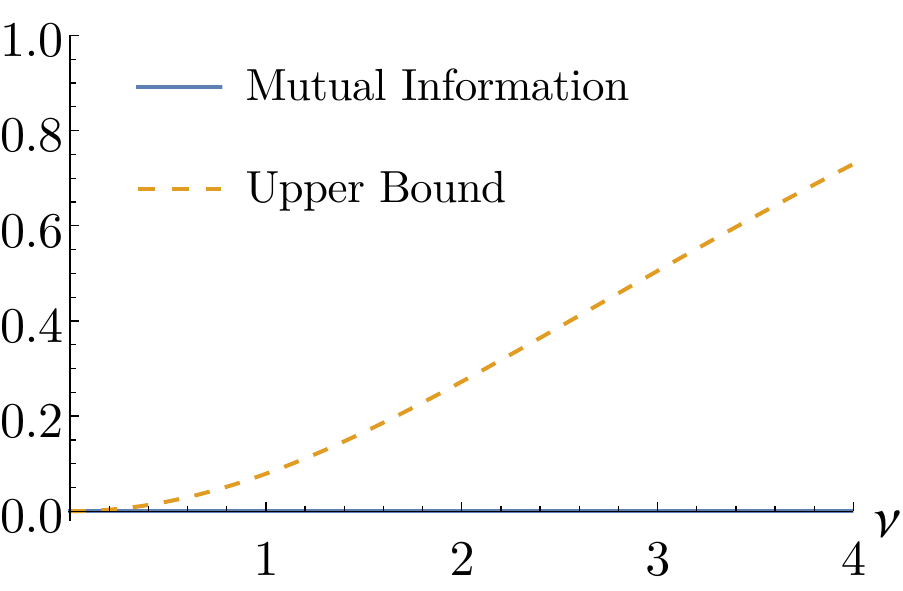}}
	\hfill
	\subfloat[Zero-mean Gaussian distributions.\label{fig:mi-gauss}]{\includegraphics[width=.45\textwidth]{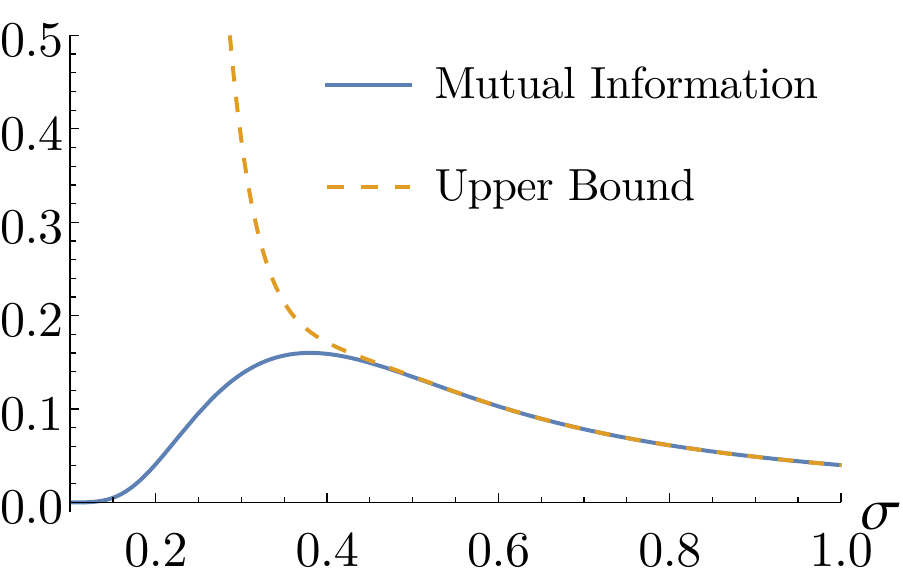}}
	\caption{Mutual information $\Info(\rv{X}_\pi; \rv{X}_\calQ)$ and its upper bound.}
	\label{fig:mutual-inf}
\end{figure}

\begin{Example}
	In the case of the exponential distributions, as in Example~\ref{ex:exponential}, the distributions of $\rv{X}_\pi$ and $\rv{X}_\calQ$ are independent, i.e., $p = p_\pi \otimes p_\calQ$, therefore $\Info(\rv{X}_\pi; \rv{X}_\calQ) = 0$. The mutual information and the corresponding upper bound \eqref{eq:mi-bound-2} are plotted in Figure~\ref{fig:mi-exp}, as function of the parameter~$\nu$.
\end{Example}

\begin{Example}
	In the case of the univariate zero-mean Gaussian distributions, as in Example~\ref{ex:gaussian}, one can use \eqref{eq:MI-decomposition} to numerically compute the mutual information $\Info(\rv{X}_\pi; \rv{X}_\calQ)$, as a function of the standard deviation $\sigma$, and compare it with the upper bound \eqref{eq:mi-bound} (Figure~\ref{fig:mi-gauss}). Interestingly, $\Info(\rv{X}_\pi; \rv{X}_\calQ)$ vanishes as $\sigma \to 0$ or $\sigma \to +\infty$, which is equivalent to choosing a lattice $\Lambda = \alpha\Z$ with $\alpha \to 0$ or $\alpha \to +\infty$, cf.~Proposition~\ref{prop:scaling}. The mutual information attains a maximum in $\sigma \approx 0.38$, showing this is the value for which $X_\pi$ and $X_\calQ$ are the least independent. 
\end{Example}

\subsection{Fisher Information}

Let $M = \set{p_\theta : \theta \in \Theta}$ be a family of probability densities $p_\theta\colon \R^n \to \R_+$ smoothly parametrized by $\theta$ in an open set $\Theta \subset \R^d$. The \emph{Fisher information matrix} is defined as the positive semi-definite matrix $G(\theta)$ with coefficients $g_{ij} (\theta) = \Exp_{p_\theta}[ \partial_i \ell_\theta \partial_j \ell_\theta ]$, where $\ell_\theta(x) \coloneqq \log p_\theta(x)$. When $M$ is a manifold satisfying certain regularity conditions~\cite{amari2000}, and $G$ is positive definite, then it becomes a Riemannian manifold with the metric given by $g_{ij}(\theta)$, called a \emph{statistical manifold}.
Let $\preceq$ denote the Loewner partial order for matrices, given by $A \preceq B$ if, and only if, $B-A$ is positive semi-definite. The following results justify the name \emph{information matrix} given to this quantity.

\begin{Proposition}[\cite{amari2000,kagan2001}]
	Let $\rv{X}$ be a random variable distributed according to a distribution parametrized by~$\theta$, and $G(\theta)$ its information matrix. The following hold.
	\begin{enumerate}
		\item \textbf{Monotonicity:} if $F\colon \calX \to \calY$ is a measurable function (i.e. a statistic) and $G_F (\theta)$ is the information matrix of $F(\rv{X})$, then $G_F(\theta) \preceq G(\theta)$, with equality if, and only if, $F$ is a sufficient statistic for $\theta$.
		\item \textbf{Additivity:} if $\rv{X}, \rv{Y}$ are independent random variables, then  the joint information matrix satisfies $G_{(X,Y)} (\theta) = G_X(\theta) + G_Y(\theta)$.
	\end{enumerate}
\end{Proposition}

Let $\rv{X}$ be a random variable on $\R^n$, and $\rv{X}_\pi$ and $\rv{X}_Q$ its wrapped and quantized factors, respectively. We denote their respective Fisher information matrices by $G(\theta)$, $G_\pi(\theta)$ and $G_\calQ(\theta)$. By additivity, the Fisher information of $p_\pi \otimes p_\calQ$ is $\tilde G (\theta) \coloneqq G_\pi (\theta) + G_\calQ (\theta)$, and, by monotonicity, we have both $G_\pi (\theta) \preceq G(\theta)$ and $G_\calQ (\theta) \preceq G(\theta)$. It follows immediately that

\begin{equation}
	\frac{\tilde G(\theta)}{2} = \frac{G_\pi (\theta) + G_Q (\theta)}{2} \preceq G(\theta).
\end{equation}

\begin{Example}
	In the family of exponential distributions, as in Example \ref{ex:exponential}, the independence of $X_\pi$ and $X_\calQ$ implies that the Fisher information matrix is additive. Indeed, for $\Lambda = \alpha \Z$:
	
	\begin{equation*}
		G(\nu) = \frac{1}{\nu^2},\quad
		G_\pi(\nu) = \frac{1}{\nu^2} + \frac{\alpha^2}{2(1-\cosh(\alpha\nu))},\quad \text{and}\quad
		G_\calQ(\nu) = \frac{\alpha^2}{2(\cosh (\alpha\nu) -1)}.
	\end{equation*}
\end{Example}

\section{A Generalization to Topological Groups}

A topological group is a topological space $(G, \tau_G)$ that is also a group with respect to some operation $\cdot$ called \emph{product}, and such that the inverse $g^{-1}$ and product $g\cdot h$ are continuous. As additional requisites, we ask $G$ to be locally compact, Hausdorff and second-countable (has a countable basis)~\cite{pontryagin1986}.
Let $\calB_G$ be the the Borel $\sigma$-algebra of $G$. Haar's theorem says there is a unique (up to a constant) Radon measure on $G$ that is invariant by left translations---we will suppose a fixed normalization, and denote both the measure and integration with respect to it by $\dif g$. The group $G$ is said to be \emph{unimodular} if $\dif g$ is also invariant by right translations. Since $G$ is $\sigma$-compact, the Haar measure is $\sigma$-finite~\cite{heyer1977}.

Let $\Gamma$ be a discrete subgroup of $G$, which is necessarily closed, since $G$ is Hausdorff, and countable, since $G$ is second-countable. Let us also consider the quotient space of left cosets $G/\Gamma = \set{ \bar g = g\Gamma \; \middle| \; g \in G },$ which has a natural projection $\pi\colon G \to G/\Gamma$, given by $\pi(g)= \bar g$. We call $\Gamma$ a \emph{lattice} if the induced Haar measure on $G/\Gamma$ is finite and bi-invariant. A particular case is when the quotient $G/\Gamma$ is compact; then $\Gamma$ is said to be a \emph{uniform lattice}. A \emph{cross-section} is defined as a set $\calD \subset G$ of representatives of $G/\Gamma$ such that all cosets are uniquely represented. A \emph{fundamental domain} is a measurable cross-section. It can be shown that $\Gamma$ is a lattice if, and only if, it admits a fundamental domain. Furthermore, every fundamental domain has the same measure~\cite{raghunathan1972, reiter2000}.

Let $P$ be a probability measure on the space $(G,\calB_G)$ that is absolutely continuous with respect to the Haar measure $\dif g$. By the Radon-Nikodym theorem, we can define a density function $p = \od{P}{g} \in L^1 (G)$, such that $p \ge 0$ and  $P(A) = \int_A p(g) \dif g$, for all $A \in \calB_G$. The original measure can be represented as $P = p \dif g$, and we consider the family of all such densities as

\begin{equation*}
	\calP (G) = \set{ p\in L^1(G) \; \middle| \; p\ge0\; \text{$\mu$-a.s.},\; \textstyle\int p \dif g = 1}.    
\end{equation*}

Probability distributions on locally compact groups have been studied in \cite{heyer1977}, and some information-theoretic properties have been investigated in~\cite{chirikjian2009,johnson2000, chirikjian2010}. The result that allows us to consider wrapped distributions in this context is the \emph{Weil formula}, taken as a particular case of \cite[Thm. 3.4.6]{reiter2000}:

\begin{Theorem} \label{thm:weil}
	For any $f \in L^1(G)$, the wrapping $f_\pi \in L^1 (G/\Gamma)$, $f_\pi (\bar g) \coloneqq \sum_{\lambda \in \Gamma} f(g\lambda)$ is well defined $\dif\bar g$-almost everywhere, belongs in $L^1(G/\Gamma)$, and
	
	\begin{equation}
		\int_G f(g) \dif g
		=
		\int_{G/\Gamma} \sum_{\lambda \in \Gamma} f(g\lambda) \dif\bar g.
	\end{equation}
	
\end{Theorem}

As a consequence, for every probability density $p \in \calP (G)$, we can consider its wrapping $p_\pi (\bar g) = \sum_{\lambda \in \Gamma} p(g \lambda)$, which is $L^1(G/\Gamma)$, non-negative and is also a probability density: $\int_{G/\Gamma} p_\pi \dif\bar g = 1$. The associated probability measure over $(G/\Gamma, \calB_{G/\Gamma})$ is $P_{\pi} = p_\pi \dif\bar g$. This notation, suggesting $P_{\pi}$ as the push-forward measure by $\pi$, is not a coincidence, since, from Theorem \ref{thm:weil},
\[    \pi_* P(A) =
\int_{G} \mathbbm{1}_{A}(\pi(g)) p(g) \dif g =
\int_{G/\Gamma} \sum_{\lambda \in \Gamma} \mathbbm{1}_{A}(\pi(g)) p(g\lambda) \dif\bar g =
\int_{G/\Gamma} \mathbbm{1}_A(\bar g) p_\pi (\bar g) \dif\bar g =
P_\pi (A).
\]

Analogously, given a fundamental domain $\calD$, it is possible to define a quantization map $\calQ \colon G \to \Lambda$ by $\calQ(g \lambda) = \lambda$, for every $g\in\calD, \lambda\in\Gamma$, which is unique since $G = \bigsqcup_{g \in \calD} g \Gamma$. The quantized probability distribution is the discrete probability measure $P_\calQ$ over $\Lambda$, defined by the mass function $p_\calQ (\lambda) = \int_{\calD} p(g\lambda) \dif g$, or as the push-forward measure $\calQ_* P$.

If $\rv{X}$ is distributed according to $p$, and $\rv{X}_\pi = \pi(\rv{X}) \sim p_\pi$, $\rv{X}_\calQ = \calQ(\rv{X}) \sim p_\calQ$, then $\rv{X} = \rv{X}_\pi \cdot \rv{X}_\calQ$, again, as a consequence of $g \mapsto \del{\pi(g), \calQ(g)}$ being a measurable bijection whose inverse is the product $\pi(g) \cdot \calQ(g)$. Despite being an abstract definition, this framework expands the scope of the previous approach, cf.~examples below. In the following, let $\Lambda \subset \R^n$ be a full-rank lattice, and $\Lambda_s \subset \Lambda$ be a full-rank sublattice, as defined in Section~\ref{sec:Lattices}.

\begin{Example}
	Let $G = \R^n$ and $\Gamma = \Lambda$. This recovers the approach from Section~\ref{sec:Lattices} as a particular case.
\end{Example}

\begin{Example}\label{ex:sublattice}
	Let $G = \Lambda$, and $\Gamma = \Lambda_s$. A fundamental domain is a choice $\calD = \set{d_1, \dots, d_k}$ of $k= \abs{\Lambda/\Lambda_s}$ points, where each point corresponds to a coset $\bar\lambda = (\lambda+\Lambda_s) \in \Lambda / \Lambda_s$. Of particular interest are Voronoi constellations~\cite{forney1989, boutros2017} where the coset leaders are selected, with some choice made for ties.
	Since $\Lambda$ is discrete, the Haar measure is the counting measure $\mu(A) = |A|$, and $p\colon \Lambda \to [0,1]$. The wrapped and quantized distributions are $p_\pi(\bar\lambda) = \sum_{\lambda_s \in \Lambda_s} p(\lambda+\lambda_s)$, and $p_\calQ (\lambda_s) = \sum_{i=1}^k p(d_i + \lambda_s)$.
\end{Example}

\begin{Example}
	Let $G = \R^n/\Lambda_s$ (a torus) and $\Gamma = \pi_s(\Lambda)$ (the projection of $\Lambda$ to $G$). Then $\pi_s(\Lambda)$ consists of a finite family of cosets $\bar\lambda_1,\dots,\bar\lambda_k$, for $k=\abs{\Lambda/\Lambda_s}$, and a choice of fundamental domain $\bar\calD$ is the projection of a fundamental domain $\calD$ of $\Lambda$.
	There are some standard choices for the distribution on $G$, such as a wrapping from the Euclidean space and the bivariate von Mises distribution~\cite[Section 11.4]{mardia2000}. Then $p_\pi(\bar x) = \sum_{i=1}^k p(\bar x + \bar \lambda_i)$ and $p_\calQ (\bar\lambda_i) = \int_{\bar\calD} p(\bar x + \bar\lambda_k) \dif \bar x$, and, in the particular case where $p(\bar x) = \sum_{\lambda_s \in \Lambda_s} p(x+\lambda_s)$ is a $\Lambda_s$-wrapped distribution, they become $p_\pi (\bar x) = \sum_{i=1}^k \sum_{\lambda_s \in \Lambda_s} p(x+\lambda_s + \lambda_i)$ and $p_\calQ(\bar\lambda_i) = \sum_{\lambda_s \in \Lambda_s} \int_{\calD} p(x+\lambda_s+\lambda_i) \dif x$.
\end{Example}

\begin{Example}
	Let $G = \mathbb{F}_q^n$ (a finite field) or $G=\Z_q^n$, and $\Gamma = \mathcal C$ (any linear block code). A fundamental domain can be a finite set of points that tiles the space by $\mathcal C$. The distributions then become finite sums, such as in Example~\ref{ex:sublattice}.
\end{Example}

\begin{Example}
	Let $G = \mathrm{SL}(n,\R)$ the Lie group of square matrices with determinant $1$, and $\Gamma=\mathrm{SL}(n,\Z)$ (the subgroup of integer matrices). This is in fact a lattice, since for $n=2$, $\vol(G/\Gamma) = \sqrt{2}\zeta(2)$ where $\zeta$ is the Riemann zeta function, and for $n>2$ the finite covolume is calculated in \cite{paula2019}, where descriptions of fundamental domains are also given.
\end{Example}

\section{Conclusion}

In this work, we have studied the decomposition of a random variable through lattices into its wrapping and quantization terms. Generalization of examples and of Proposition~\ref{prop:mi-upper-bound} to higher dimensions constitutes work in progress. We have also proposed a generalization of this decomposition to topological groups; in particular, this allows one to study information theory on such abstract spaces, which is another perspective for future work.

\section*{Acknowledgments}

This work was partly supported by Brazilian National Council for Scientific and Technological Development~(CNPq) grants 141407/2020-4 and 314441/2021-2, and by S\~{a}o Paulo Research Foundation~(FAPESP) grant 2021/04516-8. The authors are grateful to Prof. Max Costa for fruitful discussions.

\bibliographystyle{ieeetr}
\bibliography{biblio}

\end{document}